\documentclass[letterpaper,12pt,conference]{ieeeconf}  % Comment this line out
\IEEEoverridecommandlockouts                              % This command is only
\usepackage{times} % assumes new font selection scheme installed

\usepackage{amsmath}
\usepackage{amssymb}
\usepackage{amsfonts}
\usepackage{graphicx}
\usepackage{color}

\newtheorem{lemma}{Lemma}
\newtheorem{theorem}{Theorem}

\def\<{\leqslant}           % nice less than or equal to sign
\def\>{\geqslant}           % nice larger than or equal to sign
\def\d{\partial}
\def\wh{\widehat}

\def\Re{{\rm Re\,}}   % real part
\def\cH{{\cal H}}   % Hardy space

\def\mR{{\mathbb R}}    % real line
\def\mC{{\mathbb C}}    % complex plane

\def\Tr{{\rm Tr}}       % matrix trace
\def\rT{{\rm T}}        % matrix transpose
\def\bE{{\mathbf E}}    % expectation

\def\bK{{\mathbf K}}    % cumulant

\def\bra{{\langle}}
\def\ket{{\rangle}}

\def\re{{\rm e}}        % number e
\def\rd{{\rm d}}        % differential

\def\fL{{\mathfrak L}}

\def\bQ{{\mathbf Q}}

\def\bGamma{{\mathbf \Gamma}}
\def\bPi{{\mathbf \Pi}}

\def\x{\times}

\def\cF{{\cal F}}

\def\cC{{\cal C}}

\def\cA{{\cal A}}
\def\cB{{\cal B}}
\def\cE{{\mathcal E}}
\def\cov{{\bf cov}}
\def\var{{\bf var}}

\def\bL{{\mathbf L}}

\def\mT{{\mathbb T}}

\def\Ups{\Upsilon}
\def\phi{\varphi}

\begin{document}
%%%%%%%%%%%%%%%%%%%%%%%%%%%%%%%%%%%%%%%%%%%%%%%%%%%%%%%%%%%%%%%%%%%%%%%%%%%%%%%
%%%%%%%%%%%%%%%%%%%%%%%%%%%%%%%%%%%%%%%%%%%%%%%%%%%%%%%%%%%%%%%%%%%%%%%%%%%%%%%
\title{\Large\bf
Hardy-Schatten Norms of Systems, Output Energy
Cumulants and Linear Quadro-Quartic  Gaussian
    Control}
%==============================================================================
\author{
Igor G. Vladimirov,
\qquad Ian R. Petersen
\thanks{This work is supported by the Australian Research Council.
The authors are with the School of Engineering and Information
Technology, University of New South Wales at
    the  Australian  Defence Force Academy, Canberra ACT 2600,
    Australia. E-mail: 
        {\tt igor.g.vladimirov@gmail.com,
        i.r.petersen@gmail.com}.}
}
%==============================================================================
\onecolumn
\pagestyle{plain}

\maketitle
\thispagestyle{empty}
%\pagestyle{empty}

%==============================================================================
\begin{abstract}
This paper is concerned with linear stochastic control systems in state space. The integral of the squared norm of the system output over a bounded time interval is interpreted as energy. The cumulants of the output energy in the infinite-horizon limit are related to Schatten norms of the system in the Hardy space of transfer functions and the risk-sensitive performance index. We employ a novel performance criterion which seeks to minimize a combination of the average value and  the variance of the output energy of the system per unit time. The resulting  linear quadro-quartic Gaussian control problem involves the $\cH_2$ and $\cH_4$-norms of the closed-loop system. We obtain equations for the optimal controller and outline a homotopy method which reduces the solution of the problem to the numerical integration of  a differential equation initialized by the standard linear quadratic Gaussian  controller.
\end{abstract}
\section{Introduction}
%%%%%%%%%%%%%%%%%%%%%%%%%%%%%%%%%%%%%%%%%%%%%%%%%%%%%%%%%%%%%%%%%%%%%%%%%%%%%%%

This paper is concerned with linear multi-input multi-output control systems, governed in state space by Ito stochastic differential equations, driven by a standard Wiener process which is regarded  as a random disturbance. The integral of the squared Euclidean norm of the system output over a bounded time interval  is interpreted as \textit{energy}. In the disturbance attenuation paradigm, the output energy is to be minimized  in some sense.

Linear Quadratic Gaussian (LQG) control \cite{AM_1990}, for example, seeks to  minimize the expectation of the output energy which, in the infinite-horizon limit, reduces to the squared $\cH_2$-norm of the closed-loop system in an appropriate Hardy space of transfer functions. An alternative performance index is employed in the Risk-Sensitive  and Minimum Entropy control theories \cite{MG_1991}. They utilise the expected value of the exponential of the output energy multiplied by a scaling parameter to adjust the risk sensitivity.  Risk-sensitive control  extends the LQG approach and is robust with respect to  Kullback-Leibler relative entropy bounded uncertainties in the random noise \cite{DJP_2000}.

The risk-sensitive performance index can be represented as a series expansion with respect to the energy scaling parameter. The coefficients of this series are the rates of the asymptotically linear growth of the cumulants of the output energy in the infinite-horizon limit. The cumulant growth rates are directly related to higher-order Schatten norms \cite{Simon_2005} of the transfer function of the system in an appropriate Hardy space. This allows the risk-sensitive criterion to be  viewed  as a linear combination of powers of Hardy-Schatten norms of the system whose weights are governed by the risk-sensitivity parameter in a very specific way. The ``reverse engineering''  of the risk-sensitive index suggests a wide family of performance criteria in the form of linear combinations of powers of the Hardy-Schatten norms. This gives rise to a class of \textit{output energy cumulant} (OEC) control problems which extend the risk-sensitive paradigm. In fact, the LQG approach can be considered to explore this freedom to a certain degree by retaining the first term (the squared $\cH_2$-norm of the system)  of the risk-sensitive index expansion.

The present paper develops the OEC control idea,  outlined above, by employing a performance criterion which seeks to minimize a combination of the average value and  the variance of the output energy of the system per unit time. The resulting linear quadro-quartic Gaussian (LQQG) control problem utilizes a \textit{quadro-quartic functional} as a finer truncation of the risk-sensitive performance index which retains the $\cH_2$  and $\cH_4$-norms of the closed-loop system and the risk-sensitive parameter.

The $\cH_4$-norm, which involves the Schatten 4-norm of matrices \cite{HJ_2007} and is referred to as the \textit{quartic norm}, was introduced in \cite{VKS_1996}  as a subsidiary construct in the anisotropy-based robust control theory for discrete-time stochastic systems. In the present study, the quartic norm plays a central role and, in addition to providing the next  term in the risk-sensitive index expansion, quantifies (via the $\cH_4$ to $\cH_2$-norms ratio) the time scale beyond which the infinite-horizon LQG cost starts manifesting itself in sample paths of the output energy of the system.  %Minimizing the quartic norm at a given level of the $\cH_2$-norm of the system would therefore shorten the time scales needed for the infinite-horizon LQG cost to be a relevant measure of the disturbance attenuation capabilities of the system.

We consider the LQQG problem in the class of linear stabilizing controllers with the same state dimension as the underlying plant. This allows equations for an optimal controller to be obtained by using  Frechet derivatives of the quadro-quartic performance index of the closed-loop system with respect to the state-space realization matrices of the controller. The resulting set of equations depends on the risk sensitivity parameter and yields the standard LQG controller for a zero value of the parameter.
We outline a homotopy method which regards the parameter as a fictitious time variable and reduces the solution of the set of equations to a problem involving  the numerical integration of  an ordinary  differential equation (ODE) initialized by the standard LQG controller.

In addition to its possible extension to the
discrete-time case, the LQQG approach may also find application in the control of quantum
stochastic  systems as an alternative to the risk-sensitive control
paradigm.

\section{Variance of output energy and quartic norm \label{sec:H4}}
%%%%%%%%%%%%%%%%%%%%%%%%%%%%%%%%%%%%%%%%%%%%%%%%%%%%%%%%%%%%%%%%%%%%%%%%%%%%%%%

Suppose $W := (w_t)_{t\in \mR}$ is a $m$-dimensional standard Wiener
process (initialised in the infinitely distant past) at
the input of a linear time invariant (LTI) system $F$
with a square integrable $\mR^{p\x m}$-valued impulse response function $f:= (f_t)_{t\> 0}$; see Fig.~\ref{fig:zfw}.
%==============================================================================
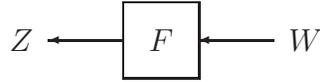
\begin{figure}[htb]
\begin{center}
\unitlength=1.0mm
%\linethickness{0.6pt}
\begin{picture}(50.00,10.00)
    \put(20,0){\framebox(10,10)[cc]{$F$}}
    \put(40,5){\vector(-1,0){10}}
    \put(20,5){\vector(-1,0){10}}
    \put(42,5){\makebox(0,0)[lc]{$W$}}
    \put(8,5){\makebox(0,0)[rc]{$Z$}}
\end{picture}
\caption{
    An LTI system $F$ with input $W$ and
    output $Z$.}
\label{fig:zfw}
\end{center}%\vskip-5mm
\end{figure}
%==============================================================================
The output $Z := (z_t)_{t\in \mR}$ of the system is a $\mR^p$-valued Gaussian random process  defined by the Ito stochastic integral
$    z_t
    :=
    \int_{-\infty}^{t}
    f_{t-s} \rd w_s
$.
 The mean
value of $Z$ is zero and the  covariance function is
\begin{equation}
\label{cov}
    c_t
    :=
    \bE(z_t z_0^{\rT})
    =
    \int_{0}^{+\infty}
    f_{s+t}f_s^{\rT}
    \rd
    s
    =
    \frac{1}{2\pi}
    \int_{-\infty}^{+\infty}
    S(\omega)
    \re^{i\omega t}
    \rd \omega
    =
    c_{-t}^{\rT},
    \qquad
    t\> 0,
\end{equation}
where
\begin{equation}
\label{S}
    S(\omega)
    :=
    \wh{F}(\omega) \wh{F}(\omega)^*
    =
    \int_{-\infty}^{+\infty}
    c_t
    \re^{-i \omega t}\
    \rd t
\end{equation}
is the spectral density of $Z$. Here, $(\cdot)^* :=
(\overline{(\cdot)})^{\rT}$ denotes the complex conjugate transpose of a
matrix, and
$    \wh{F}(\omega)
    :=
    F(i\omega)
    =
    \int_{0}^{+\infty}
    f_t
    \re^{-i \omega t}
    \rd t
$
is the Fourier transform of the impulse response, that is, the
boundary value of the transfer function of the system
$    F(v)
    :=
    \int_{0}^{+\infty}
    f_t
    \re^{-vt}
    \rd t
$, with $
    \Re v >0
$.
With $f$ assumed to be square integrable, $F$ belongs to the Hardy space
$\cH_2^{p\x m}$ of $\mC^{p\x m}$-valued functions of a complex
variable, analytic in the right half-plane and endowed with the
$\cH_2$-norm
\begin{equation}
\label{H2}
    \|F\|_2
    :=
    \sqrt{
        \int_{0}^{+\infty}
        \|f_t\|^2
        \rd t
    }
%    =
%    \sqrt{
%        \Tr c_0
%    }
    =
    \sqrt{
        \frac{1}{2\pi}
        \int_{-\infty}^{+\infty}
        \|\wh{F}(\omega)\|^2
        \rd
    \omega}.
\end{equation}
Here, the Plancherel theorem is used, and  $\|M\|:= \sqrt{\bra M,M\ket}$ denotes the Frobenius norm of a matrix $M$ generated by the inner product $\bra M, N\ket:= \Tr (M^*N)$, so that
$\|\wh{F}(\omega)\|^2 = \Tr S(\omega)$ is the trace of the spectral
density from (\ref{S}). In view of (\ref{cov}), $\|F\|_2^2 = \Tr c_0 =
\bE(|z_t|^2)$ is the variance of the output signal for any $t$.
For a finite
time horizon $T>0$, the random variable
\begin{equation}
\label{ET}
    \cE_T
    :=
    \int_{0}^{T}
    |z_t|^2
    \rd t
\end{equation}
is interpreted as the {\it output energy} of the system $F$ over the
time interval $[0,T]$, and
\begin{equation}
\label{eT}
    \epsilon_T
    :=
    \cE_T/T
\end{equation}
is the corresponding output energy rate. The
mean value of $\epsilon_T$ coincides with the squared $\cH_2$-norm of the system (\ref{H2}):
$
    \bE \epsilon_T
    =
    \|F\|_2^2
$.
This ensemble average can manifest itself
in sample paths of $\epsilon_T$ only by virtue of the law
of large numbers, provided $T$ is large
enough. Under additional assumptions on the system
$F$, the rate of the mean square convergence
$    \mathop{\rm l.i.m.}_{T\to +\infty}
    \epsilon_T
    =
    \| F\|_2^2
$ is quantified by the asymptotic
behaviour of the variance of $\epsilon_T$. The convergence rate is described by the lemma below in terms of the quantity
\begin{equation}
\label{H4}
    \|F\|_4
     :=
    \sqrt[4]{
    \frac{1}{2\pi}
    \int_{-\infty}^{+\infty}
    \|S(\omega)\|^2
    \rd \omega}
%    =
%    \sqrt[4]{
%    \int_{-\infty}^{\infty}
%    \|c_t\|^2
%    \rd t
%    }
    =
    \sqrt[4]{
    2
    \int_{0}^{+\infty}
    \|c_t\|^2
    \rd t
    }.
\end{equation}
This is a
continuous-time counterpart of the $\cH_4$-norm introduced
 as a subsidiary construct in the anisotropy-based
robust control of discrete-time systems \cite{VKS_1996}. The second equality in (\ref{H4}) follows from the Plancherel theorem
applied to the spectral density (\ref{S}). The systems $F$ with
$\|F\|_4<+\infty$ form a normed space $\cH_4^{p\x m}$. The
integrand $\|S(\omega)\|^2 = \Tr ((\wh{F}(\omega) \wh{F}(\omega)^*)^2)$ in (\ref{H4}) is the
fourth power of the Schatten 4-norm \cite[p.~441]{HJ_2007} of the
matrix $\wh{F}(\omega)$; see also \cite{Simon_2005}. The
$\cH_4$-norm $\|F\|_4$ will be referred to as the \textit{quartic norm} of the system $F$.

%Note that (\ref{H4}) is slightly different from the straightforward
%extension of the conventional norm in the Hardy space $\cH_4$ of
%scalar-valued functions to the matrix case
%$$
%    \sqrt[4]{
%    \frac{1}{2\pi}
%    \int_{-\pi}^{\pi}
%    \|\wh{F}(\omega)\|^4
%    \rd \omega}.
%$$
%Although both definitions coincide in the case, where at least one
%of the dimensions $p$ or $m$ equals 1, and yield topologically
%equivalent norms in the general MIMO setting by the inequalities
%$$
%    \frac{(\Tr \Sigma)^2}{r}
%    \<
%    \|\Sigma\|^2
%    \<
%    (\Tr \Sigma)^2
%$$
%for any positive semi-definite Hermitian matrix $\Sigma$ of rank
%$r$, the definition (\ref{H4}) is preferable due to the
%probabilistic interpretation below.

%====================================================================
\begin{lemma}
\label{lem:H4} Let $F \in \cH_2^{p\x m}\bigcap \cH_4^{p\x m}$. Then the variance of the
output energy rate (\ref{eT}) of the system behaves asymptotically
as%\vskip-5mm
\begin{equation}
\label{vareTasy}
    \var(\epsilon_T)
    \sim
    2 \|F\|_4^4/T,
    \qquad
    T\to +\infty.
\end{equation}
\end{lemma}
%====================================================================

\begin{proof}
 By applying
Lemma~\ref{lem:bilinear} of Appendix~A to the Gaussian random
vectors $z_s$ and $z_t$ and using (\ref{cov}), it follows that
$
    \cov(|z_s|^2, |z_t|^2)
    =
    2
    \|c_{s-t}\|^2
$.
Hence, the variance of the output energy (\ref{ET})  can be computed as
\begin{equation}
\label{varET}
    \var(\cE_T)
     =
    \int_{[0,T]^2}
    \cov(|z_s|^2, |z_t|^2)
    \rd s\rd t
     =
    2
    \int_{[0,T]^2}
    \|c_{s-t}\|^2
    \rd s\rd t
    =
    4T
        \int_{0}^{T}
        (
            1-u/T
        )
        \|c_u\|^2
        \rd u,
%        \\
%
%    & \sim &
%    4T
%        \int_{0}^{+\infty}
%        \|c_u\|^2
%        \rd u,
%    \quad
%    T \to +\infty.
\end{equation}
where use is made of the property $c_t=c_{-t}^{\rT}$ and the
invariance of the Frobenius norm of a matrix under the transpose.
Since the assumption $F \in \cH_4^{p\x m}$ ensures the square
integrability of the covariance function (\ref{cov}), then
\begin{equation}
\label{varET1}
    \lim_{T\to +\infty}
        \int_{0}^{T}
        (
            1-u/T
        )
        \|c_u\|^2
        \rd u
        =
        \int_{0}^{+\infty}
        \|c_u\|^2
        \rd u
\end{equation}
holds by Lebesgue's dominated convergence
theorem. Since $
    \d_T
        \int_{0}^{T}
        (
            1-u/T
        )
        \|c_u\|^2
        \rd u
        =
        T^{-2}
        \int_{0}^{T}u\|c_u\|^2\rd u
        \> 0
$, the convergence is monotonic. Now,
(\ref{vareTasy}) is obtained by using (\ref{eT}) and combining
(\ref{varET}) and (\ref{varET1}) with (\ref{H4}):
$
    \var(\epsilon_T)
    =
    \var(\cE_T)/T^2
    \sim
    4
        \int_{0}^{+\infty}
        \|c_u\|^2
        \rd u/T
        =
        2\|F\|_4^4/T
$ as $T\to +\infty$.
\end{proof}
%==============================================================================

%By combining (\ref{outputenergyrate}) and (\ref{limeT}) with
%(\ref{vareTasy}), it follows that
%\begin{equation}
%\label{varbarET}
%    \bE((\epsilon_T - \|F\|_2^2)^2)
%    =
%    \var(\epsilon_T)
%    \sim
%    \frac{2\|F\|_4^4}{T},
%    \quad
%    T \to +\infty.
%\end{equation}

In view of a central limit theorem for quadratic functionals of Gaussian processes
\cite[Theorem~2]{Ginovian_1994}, the relation (\ref{vareTasy}) provides the scaling factor for the asymptotic
standard normality of the random variable $\sqrt{T/2} (\epsilon_T - \|F\|_2^2)/\|F\|_4^2$ as $T\to +\infty$. Heuristically,  the root mean square deviation of $\epsilon_T$ from its mean
value $\|F\|_2^2$ is relatively small if
\begin{equation}
\label{T*}
    T
    \gg
    T_*
    :=
    2
    (
        \|F\|_4/\|F\|_2
    )^4.
\end{equation}
The right-hand side of (\ref{T*}) quantifies the time horizon beyond
which the $\cH_2$-norm $\|F\|_2$ manifests itself in the sample
paths of the output energy of the system. On the other hand, for $T
\ll T_*$, the ergodic properties of the system output $Z$ do not
expose themselves since the expected value $\bE \epsilon_T =
\|F\|_2^2$ of the output energy rate is ``indistinguishable'' in the
background of random fluctuations whose standard deviation can be estimated by using
(\ref{vareTasy}) as
$
    \sqrt{\var(\epsilon_T)}
    \sim
    \|F\|_4^2\,
    \sqrt{2/T}
    \gg
    \bE \epsilon_T
$. Thus, the squared $\cH_2$-norm as the
average output energy loses its significance for quantifying the disturbance attenuation capabilities of the system on short time scales $T\ll T_*$. The critical time horizon $T_*$ defined by (\ref{T*})
is similar to the {\it integral time scale} of measurements  in
turbulent flows \cite[pp.~50--51]{Frisch_1995}. As an example, let $Z$ be an Ornstein-Uhlenbeck
process generated from a standard Wiener process $W$ by a single-input single-output system $F$ according to the SDE
\begin{equation}
\label{OU}
    \rd z_t = az_t\rd t + \sqrt{2|a|}\,\rd w_t,
\end{equation}
parameterized by $a <0$.
The covariance function (\ref{cov}) of $Z$ is
$
    c_t
    =
    \re^{a|t|}
$, and
the $\cH_2$
and $\cH_4$-norms of the system $F$, defined by (\ref{H2}) and (\ref{H4}),  are
$
    \|F\|_2 = 1
$
and
$
    \|F\|_4
    =
    |a|^{-1/4}
$.
Therefore, the critical time
horizon (\ref{T*}) takes the form  $    T_*
    =
    2/|a|
$  and coincides with the typical transient time of the process; see Fig.~\ref{fig:energy}.
%==============================================================================
\begin{figure}[htb]
\begin{center}
\includegraphics[width=80mm]{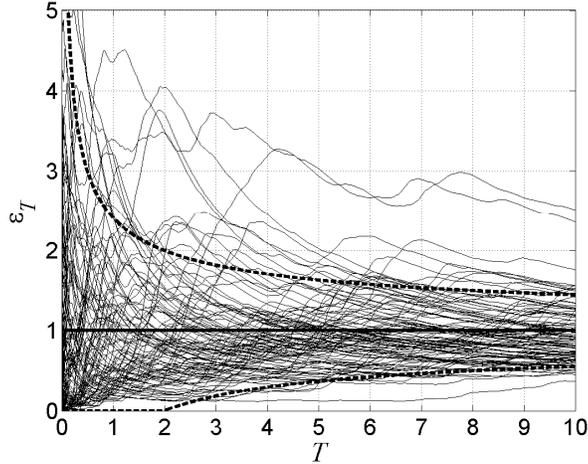}
\caption{
    100 sample paths of $\epsilon_T$ versus $T\<10$ for the Ornstein-Uhlenbeck process
    generated by  (\ref{OU}) with $a=-1$, so that the critical time horizon
     beyond which $\epsilon_T$ exposes relative proximity
    to the limit value $\|F\|_2^2 = 1$ (horizontal bold line) is $T_*= 2$.
    The dashed bold lines localize the typical values of $\epsilon_T$ which form a ``tube'' of half-width
    $\sqrt{T_*/T}$ about the limit.
}
    \label{fig:energy}
\end{center}
\end{figure}
%==============================================================================

%%%%%%%%%%%%%%%%%%%%%%%%%%%%%%%%%%%%%%%%%%%%%%%%%%%%%%%%%%%%%%%%%%%%%%%%%%%%%%%
\section{Cumulants of output energy and Hardy-Schatten norms\label{sec:cumulants}}
%%%%%%%%%%%%%%%%%%%%%%%%%%%%%%%%%%%%%%%%%%%%%%%%%%%%%%%%%%%%%%%%%%%%%%%%%%%%%%%

For a finite time horizon $T>0$, let $C_T$ denote a Toeplitz  integral
operator whose kernel is specified by the covariance function
(\ref{cov}). An $\mR^p$-valued integrable function $\psi:=
(\psi_t)_{0\<t \<T}$ is mapped by $C_T$ to $\phi:= (\phi_s)_{0\< s\< T}$ as
$    \phi_s
    :=
    \int_{0}^{T}
    c_{s-t}\psi_t
    \rd
    t
$.
Suppose $\theta$ is a real parameter
satisfying $0<\theta < 1/\rho(C_T)$, where $\rho(\cdot)$ is the spectral radius. In view of the Fredholm
formula \cite[Theorem~3.10 on p.~36]{Simon_2005} (see also \cite{Ginovian_1994} and references therein),
\begin{equation}
\label{expquad}
    \ln \bE
    \re^{\theta\cE_T/2}
     =
    -
    \frac{1}{2}
    \Tr
    \ln(
        I-\theta C_T
    )
     =
        \frac{1}{2}
        \sum_{k\> 1}
        \theta^k
        \Tr(C_T^k)/k,
\end{equation}
where $I$ is the identity operator. The trace of the $k$-fold iterate
of $C_T$ is computed as
\begin{equation}
\label{TrCT}
    \Tr(C_T^k)
    =
    \int_{[0,T]^k}
    \Tr
    (
        c_{t_0-t_1}
        c_{t_1-t_2}
        \x
        \ldots
        \x
        c_{t_{k-2}-t_{k-1}}
        c_{t_{k-1}-t_0}
    )
    \rd t_0
    \x \ldots \x
    \rd t_{k-1}.
\end{equation}
The expectation in (\ref{expquad}) is the
moment-generating function of $\cE_T$, and
hence,
\begin{equation}
\label{expquad1}
    \ln \bE\re^{\theta\cE_T/2}
    =
    \sum_{k\> 1}
    (\theta/2)^k
    \bK_k(\cE_T)/k!
    =
    \theta
    \left(
        \bE\cE_T
        +
        \theta
        \var(\cE_T)/4
    \right)/2
    +
    O(\theta^3),
    \qquad
    \theta \to 0.
\end{equation}
Here,
$
    \bK_k(\xi)
    :=
    \left.
        \d_v^k
        \ln \bE \re^{v\xi}
    \right|_{v=0}
    =
    P_k(\bE \xi, \ldots, \bE (\xi^k))
$
denotes the $k$th cumulant of a random variable $\xi$, which is
related with the first $k$ moments of $\xi$ via a universal
polynomial $P_k$. The first three of these
polynomials are
$
P_1(\mu_1)
=
    \mu_1
$,
$
P_2(\mu_1, \mu_2)
 =
    \mu_2 - \mu_1^2
$ and
$
P_3(\mu_1, \mu_2, \mu_3)
  =
    \mu_3 - 3\mu_1\mu_2 + 2\mu_1^3$.
By comparing the power series in (\ref{expquad}) and
(\ref{expquad1}) and using the identity $(2r)!! = r!2^r$, it follows that
the $k$th cumulant of the output energy (\ref{ET})
of the system is related to the trace (\ref{TrCT}) as
\begin{equation}
\label{Kzeta}
    \bK_k(\cE_T)
    =
    (2k-2)!!
    \Tr(C_T^k).
\end{equation}
Using (\ref{S}) and extending (\ref{H2}) and (\ref{H4}),
we define, for a positive  integer $k$, a
higher order Hardy norm of the system $F$ by
\begin{equation}
\label{H2k}
    \|F\|_{2k}
     :=
    \sqrt[2k]{
    \frac{1}{2\pi}
    \int_{-\infty}^{+\infty}
    \Tr(S(\omega)^k)
    \rd \omega},
\end{equation}
which reproduces the $\cH_2$
and $\cH_4$-norms for $k=1,2$. Here, $\sqrt[2k]{\Tr(S(\omega)^k)}$
is the Schatten $2k$-norm \cite[p.~441]{HJ_2007} of the matrix
$\wh{F}(\omega)$. The resulting  Hardy-Schatten
space $\cH_{2k}^{p\x m}$ is equipped with the
norm $\|\cdot\|_{2k}$.
%==============================================================================
Similarly to the $\cH_2$-norm, the $\cH_{2k}$-norms (\ref{H2k}) are all
invariant under replacing the system $F$ with its dual $F^{\dagger}$,
\begin{equation}
\label{dual_invariance}
    \|F^{\dagger}\|_{2k}
    =
    \|F\|_{2k},
    \qquad
    k \> 1,
\end{equation}
where $F^{\dagger}$ has the transposed impulse response
$(f_t^{\rT})_{t \> 0}$. Indeed, the transpose of a square matrix
does not modify its spectrum, and for conformable complex
matrices $X$ and $Y$, the matrices $XY$ and
$YX$ share nonzero eigenvalues. Therefore, with the
dependence  on the frequency $\omega$ omitted for brevity,
$
    \Tr
    (
        (
            \wh{F}^{\rT}
            (
                \wh{F}^{\rT}
            )^*
        )^k
    )
    =
    \Tr
    (
        (
            (
                \wh{F}^*
                \wh{F}
            )^k
        )^{\rT}
    )
    =
    \Tr
    (
        (
            \wh{F}
            \wh{F}^*
        )^k
    )
$,
and hence (\ref{dual_invariance}) follows.
%, so that they form a decreasing sequence
%$$
%    \cH_{\infty}^{p\x m}
%    \subset
%    \cH_{2k+2}^{p\x m}
%    \subset
%    \cH_{2k}^{p\x m} \subset \cH_2^{p\x m},
%    \quad
%    k> 1.
%$$
%Moreover, the norms (\ref{H2k}) are convergent  to the
%$\cH_{\infty}$-norm in the sense that for any $F \in
%\cH_{\infty}^{p\x m}$,
%$$
%    \lim_{k \to +\infty}
%    \|F\|_{2k}
%    =
%    \|F\|_{\infty}
%    :=
%    \mathop{\rm ess\, sup}_{-\pi \< \omega \< \pi}
%    \sigma_{\rm max}(\wh{F}(\omega)),
%$$
%where $\sigma_{\rm max}$ is the largest singular value of a matrix.
By the Szeg\H{o} limit theorem for Toeplitz operators
\cite{GS_1958}, under additional integrability conditions,
\begin{equation}
\label{Szego}
    \lim_{T \to +\infty}
    \frac{\Tr \chi(C_T)}{T}
    =
    \frac{1}{2\pi}
    \int_{-\infty}^{+\infty}
    \Tr \chi(S(\omega)) \rd \omega.
\end{equation}
Here, $\chi$ is a function of
a complex variable, satisfying $\chi(0) = 0$ and analytic in a neighbourhood of the
interval $[0,\|F\|_{\infty}^2]$, with $\|F\|_{\infty}$ the $\cH_{\infty}$-norm of $F$.
In view of (\ref{Kzeta}), the application of (\ref{Szego}) to elementary polynomials $\chi(v):=v^k$ yields
the asymptotically linear growth of the output energy cumulants with respect to time:
$    \lim_{T\to +\infty}
    (\bK_k(\cE_T)/T)
    =
    (2k-2)!!
    \|F\|_{2k}^{2k}
$,
provided $F\in \bigcap_{j=1}^{k} \cH_{2j}^{p\x m}$, with
Lemma~\ref{lem:H4} being a particular case for $k=2$.  The application of  (\ref{Szego})
to $\chi(v):= (2/\theta)\ln(1-\theta v)$, with
$0< \theta< \|F\|_{\infty}^{-2}$, gives
\begin{align}
\nonumber
    \frac{2}{\theta}\lim_{T\to +\infty}
        \frac{\ln \bE\re^{\theta\cE_T/2}}{T}
     & =
    -\frac{1}{2\pi\theta}
    \int_{-\infty}^{+\infty}
    \ln\det(I_p-\theta S(\omega))
    \rd \omega\\
\label{quartic_correction}
    & =
    \sum_{k\> 1}
    \theta^{k-1}
    \|F\|_{2k}^{2k}/k
     =
    \bQ_{\theta}(F)
    +
    O(\theta^2),
    \qquad
    \theta \to 0+,
\end{align}
where $I_p$ denotes  the identity matrix of order $p$, and
\begin{equation}
\label{QQ}
    \bQ_{\theta}(F)
    :=
    \|F\|_2^2
    +
    \theta\|F\|_4^4\big/2.
\end{equation}
The expected exponential-of-quadratic functional $\bE\re^{\theta\cE_T/2}$ in (\ref{quartic_correction}) is used as a
performance criterion in the risk-sensitive  and minimum entropy
control theories \cite{MG_1991}. The \textit{quartic norm} $\|F\|_4$
provides the next correction to the squared $\cH_2$-norm $\|F\|_2^2$ in the
series expansion (\ref{quartic_correction})  for small $\theta$.
Therefore, the {\it quadro-quartic} functional $\bQ_{\theta}$, defined by
(\ref{QQ}), can be regarded as a finer truncation of the risk-sensitive performance
index.

%==============================================================================
%We have also used the
%property that the Frobenius norm of a Hermitian matrix $\Sigma$
%satisfies $\|\Sigma\|^2 = \Tr(\Sigma^2)$, where the latter quantity
%is the sum of squared eigenvalues of $\Sigma$.

%$$
%    E_k = \frac{\rd^k}{\rd \theta^k} \bE\re^{\theta|\xi|^2}
%$$

%%%%%%%%%%%%%%%%%%%%%%%%%%%%%%%%%%%%%%%%%%%%%%%%%%%%%%%%%%%%%%%%%%%%%%%%%%%%%%%
\section{Quadro-quartic functional in state space}
%%%%%%%%%%%%%%%%%%%%%%%%%%%%%%%%%%%%%%%%%%%%%%%%%%%%%%%%%%%%%%%%%%%%%%%%%%%%%%%

Let $F$ be a strictly proper LTI system with an $m$-dimensional standard Wiener process
$W$ at the input, $p$-dimensional
output $Z$ and $n$-dimensional state $X$ governed by an Ito SDE:
\begin{equation}
\label{Fstate}
    \rd x_t
     =
     Ax_t\rd t + B\rd w_t,
    \qquad
    z_t
     =
    Cx_t,
\end{equation}
where $A\in \mR^{n\x  n}$,
$B\in \mR^{n\x m}$, $C\in \mR^{p\x n}$ are constant matrices. The state-space
representation  will be written as
%\vskip-5mm
\begin{equation}
\label{FABC}
    F
    =
    (A,B,C)
    =
        \begin{array}{rl}
         &   \!\!\!\!\!\!{}_{\leftarrow n \rightarrow}{}_{\leftarrow m \rightarrow} \\
            \begin{array}{cc}
                {}^n & \hskip-2mm\updownarrow \\
                {}_p & \hskip-2mm\updownarrow
            \end{array} &\hskip-4mm
            \left[
                \begin{array}{c|c}
                    A & B \\
                    \hline
                    C & 0 \\
                \end{array}
            \right]\\
            {}
        \end{array},
\end{equation}
where we have also shown the dimensions, and the horizontal and vertical separators serve to avoid
confusion with an ordinary block matrix.
%; see Fig.~\ref{fig:abcd}.
%%==============================================================================
%\begin{figure}[htb]
%\begin{center}
%\includegraphics[width=6cm]{Images/abcd.ps}
%\end{center}
%\caption{
%    \label{fig:abcd}
%    The system (\ref{Fstate}) with input $W$,
%    output $Z$ and state $X$.
%}
%\end{figure}
%%==============================================================================
The dual system is $
    F^{\dagger}
    =
    (A^{\rT}, C^{\rT},B^{\rT})
$. If the matrix $A$ is Hurwitz, then the mutually dual
controllability and observability Gramians $P$ and $Q$  of (\ref{FABC}) are
unique solutions of the algebraic Lyapunov
 equations
\begin{equation}
\label{PQ}
    AP + PA^{\rT} + BB^{\rT} = 0,
    \qquad
    A^{\rT}Q + Q A  + C^{\rT}C = 0.
\end{equation}
In what follows, an important role is played by the matrix
\begin{equation}
\label{H}
    H:= QP,
\end{equation}
whose spectrum is formed by the squared Hankel singular values of the system (\ref{FABC}).
We will write $\|X\|_M:= \sqrt{\Tr(X^{\rT} M X)}$ for the weighted Frobenius (semi-) norm of a
real matrix $X$ generated by a positive (semi-) definite matrix $M$.

%==============================================================================
\begin{lemma}
\label{lem:H4state} Let $F$ be an asymptotically stable system with
the  state-space realization (\ref{FABC}).  Then the quartic norm (\ref{H4})
 is expressed in terms of the Gramians $P$, $Q$ from (\ref{PQ}) and the matrix $H$ from (\ref{H}) as
\begin{equation}
\label{H4state}
    \|F\|_4^4
     =
    2
    \|(A,PC^{\rT},C)\|_2^2
%    =
%    \left\|
%        \left[
%            \begin{array}{c|c}
%            A & PC^{\rT}\\
%            \hline
%            C & 0
%            \end{array}
%        \right]
%    \right\|_2^2
    =
    2\|PC^{\rT}\|_Q^2
    =
    2
    \|(A,B,B^{\rT}Q)\|_2^2
%    \left\|
%        \left[
%            \begin{array}{c|c}
%            A & B\\
%            \hline
%            B^{\rT}Q& 0
%            \end{array}
%        \right]
%    \right\|_2^2
    =
    2\|QB\|_P^2
      =
    -4\Tr(A^{\rT}H^2).
\end{equation}
\end{lemma}
%==============================================================================

\begin{proof}
Let $Z$ be
a stationary Gaussian random process generated by
(\ref{Fstate}), with $W$ a standard Wiener process.  Then the
steady-state covariance function (\ref{cov}) is
\begin{equation}
\label{ct}
    c_t
    =
    C\re^{At} P C^{\rT},
    \qquad
    t\> 0.
\end{equation}
Here, we use the fact that the controllability
Gramian is the steady-state covariance matrix of the state of the system:
$P=\cov(x_t)$. Since the function
$c_t$ in (\ref{ct}) coincides with the impulse response
of the system $(A,PC^{\rT}, C)$, then (\ref{H4}) yields
$    \|F\|_4^4
     =
    2
    \|(A,PC^{\rT},C)\|_2^2
%    \left\|
%        \left[
%            \begin{array}{c|c}
%            A & PC^{\rT}\\
%            \hline
%            C & 0
%            \end{array}
%        \right]
%    \right\|_2^2
    =
    2
    \Tr
    (
        CPQPC^{\rT}
    )
    =
    2\|PC^{\rT}\|_Q^2
$,
which  proves the first two equalities in (\ref{H4state}). Here,
we have also used the property that the system $(A,PC^{\rT},C)$  shares the matrices $A$, $C$ with the underlying
system (\ref{FABC}) and hence, inherits from $F$ the observability
Gramian $Q$. The remaining three equalities in
(\ref{H4state}) follow from the first two by the invariance of the
$\cH_2$ and $\cH_4$-norms under taking the dual of a system, and by  the duality
of the controllability and observability Gramians.
\end{proof}

The controllability and observability Gramians $\Phi$, $\Psi$ of a subsidiary system $(A,PC^{\rT}, B^{\rT}Q)$, which satisfy the algebraic Lyapunov equations
\begin{equation}
\label{Schattenian_Phi_Psi}
    A\Phi + \Phi A^{\rT} +  P C^{\rT}C  P
      =
    0,\qquad
    A^{\rT} \Psi
    +
    \Psi A +  Q BB^{\rT}  Q
      =  0,
\end{equation}
 will be referred to as the controllability and observability \textit{Schattenians} of the system (\ref{FABC}). The representations (\ref{H4state}) imply that
% \begin{equation}
% \label{Schattenian_quartic}
$$    \|F\|_4^4
    =
    2\Tr(C \Phi C^{\rT})
    =
    2\Tr(B^{\rT} \Psi B),
$$
% \end{equation}
and hence, the significance of the Schattenians $\Phi$, $\Psi$ for the quartic norm is analogous to the role which the Gramians $P$, $Q$ play for the $\cH_2$-norm.
%The following theorem is obtained by combining Lemma~\ref{lem:H4state} with the state-space representation of the $\cH_2$-norm.

%==============================================================================
\begin{theorem}
\label{th:QQstate} Let $F$ be an asymptotically stable system with
the  state-space realization (\ref{FABC}).  Then the quadro-quartic
functional (\ref{QQ}) is expressed in terms of the Gramians  $P$, $Q$  from (\ref{PQ}) and the matrix $H$ from (\ref{H}) as
\begin{align}
\nonumber
    \bQ_{\theta}(F)
    & =
    \left\|
        (
            A,
            \begin{bmatrix}
            B & \sqrt{\theta} PC^{\rT}
            \end{bmatrix},
            C
        )
    \right\|_2^2\\
\nonumber
%    \left\|
%        \left[
%            \begin{array}{c|cc}
%            A & B & \sqrt{\theta} PC^{\rT}\\
%            \hline
%            C & 0 & 0
%            \end{array}
%        \right]
%    \right\|_2^2\\
    & =
    \Tr
    (
        (
            BB^{\rT} + \theta PC^{\rT}CP
        )
        Q
    )\\
\nonumber
    & =
    \left\|
        \left(
            A,
            B,
            \begin{bmatrix}
                C \\
\nonumber
                 \sqrt{\theta} B^{\rT}Q
            \end{bmatrix}
        \right)
    \right\|_2^2\\
%    \left\|
%        \left[
%            \begin{array}{c|c}
%            A & B\\
%            \hline
%            C & 0\\
%            \sqrt{\theta} B^{\rT}Q & 0
%            \end{array}
%        \right]
%    \right\|_2^2\\
    \label{QQstate}
    & =
    \Tr
    (
        (
            C^{\rT}C + \theta QBB^{\rT}Q
        )
        P
    )
    =
    -
    2\Tr(A^{\rT}H(I_n+\theta H)).
\end{align}
\end{theorem}
%==============================================================================
\begin{proof}
Substitution of
$\|F\|_2 = \sqrt{\Tr (B^{\rT}QB)}$ and the first two equalities from
(\ref{H4state}) into (\ref{QQ}) yields
\begin{align*}
    \bQ_{\theta}(F)
    & =
    \|(A,B,C)\|_2^2
    +
    \theta
    \|(A,PC^{\rT},C)\|_2^2\\
    & =
    \|
        (
            A,
                \begin{bmatrix}
                    B & \sqrt{\theta}PC^{\rT}
                \end{bmatrix},
            C
        )
    \|_2^2
%    =
%    \Tr
%    (
%        BB^{\rT}Q
%    )
%    +
%    \theta
%    \Tr
%    (
%        CPQPC^{\rT}
%    )\\
    =
    \Tr
    (
        (BB^{\rT} + \theta PC^{\rT}CP)
        Q
    ),
\end{align*}
which establishes the first two equalities in (\ref{QQstate}).
 The third and fourth representations of the quadro-quartic
functional are obtained from the first two by the duality argument
or directly from the third and fourth equalities in (\ref{H4state}).
The last representation of $\bQ_{\theta}(F)$ in
(\ref{QQstate}) follows from the previous ones by using the Lyapunov
equations (\ref{PQ}):
\begin{align*}
    \bQ_{\theta}(F)
    & =
    \Tr
    (
        (BB^{\rT} + \theta PC^{\rT}CP)
        Q
    )\\
    & =
    -
    \Tr
    (
        (AP+PA^{\rT} + \theta P(A^{\rT}Q + QA)P)
        Q
    )\\
    & =
    -2
    \Tr
    (
        A^{\rT}QP
        +
        \theta
        A^{\rT}(QP)^2
    )
    =
    -2
    \Tr
    (
        A^{\rT}H
        (
            I_n
            +
            \theta
            H
        )
    ).
\end{align*}
\end{proof}
%==============================================================================

%\textit{Remark~3:} The arithmetic mean of the traces from (\ref{QQstate})
%provides a ``symmetrized'' representation  of the  quadro-quartic
%functional:
%\begin{equation}
%\label{symQQstate}
%    \bQ_{\theta}(F)
%    =
%    \frac{1}{2}
%    \Tr
%    (
%        C^{\rT}CP
%        +
%        BB^{\rT}Q
%        +
%        \theta
%        (
%            C^{\rT}CPQP
%            +
%            BB^{\rT}QPQ
%        )
%    ).
%\end{equation}
%

%%%%%%%%%%%%%%%%%%%%%%%%%%%%%%%%%%%%%%%%%%%%%%%%%%%%%%%%%%%%%%%%%%%%%%%%%%%%%%%
\section{Linear quadro-quartic Gaussian control problem}
%%%%%%%%%%%%%%%%%%%%%%%%%%%%%%%%%%%%%%%%%%%%%%%%%%%%%%%%%%%%%%%%%%%%%%%%%%%%%%%

Consider a plant with an  $m_1$-dimensional standard Wiener
process $W$ as the input disturbance and an $m_2$-dimensional input
control signal $U$. The outputs of the system are a
$p_1$-dimensional to-be-controlled signal $Z$ and a
$p_2$-dimensional observation signal $Y$.
Also, the system has an $n$-dimensional state $X$. These processes are governed by
\begin{eqnarray}
\label{nextx}
    \rd x_t
    & = &
    A x_t\rd t \ + \ B_1 \rd w_t + B_2 u_t\rd t,\\
\label{z}
    z_t
    & = &
    C_1 x_t \ \ \, + \qquad \quad\quad \ D_{12} u_t,\\
\label{y}
    \rd y_t
    & = &
    C_2 x_t\rd t + D_{21} \rd w_t.
\end{eqnarray}
Here,
$
    A\in \mR^{n\x n}$,
$
    B_k\in \mR^{n\x m_k}$,
$
    C_j\in \mR^{p_j\x n}
$,
$
    D_{jk}\in \mR^{p_j\x m_k}
$, with $D_{11}=0$ and
$D_{22}=0$. The control signal $U$ is generated at the output of a
controller $K$ with input $Y$. We consider a strictly proper LTI
controller
\begin{equation}
\label{K}
    K
    =
        \begin{array}{rl}
         &   \!\!\!\!\!{}_{\leftarrow n \rightarrow}{}_{\leftarrow p_2 \rightarrow} \\
            \begin{array}{cc}
                {}^n & \hskip-3mm\updownarrow \\
                {}_{m_2} & \hskip-3mm\updownarrow
            \end{array} &\hskip-4mm
            \left[
                \begin{array}{c|c}
                    a \ & \ b \, \\
                    \hline
                    c \ & \ 0 \,\\
                \end{array}
            \right]\\
            {}
        \end{array},
\end{equation}
with an $n$-dimensional state $\Xi$. It is driven by the observation
$Y$ and produces the output $U$ as
\begin{equation}
\label{y_to_xi_yxi_to_u}
    \rd \xi_t
     =
    a\xi_t\rd t + b\rd y_t,
    \qquad
    u_t
     =
    c\xi_t,
\end{equation}
where $a\in \mR^{n\x n}$, $b\in \mR^{n\x p_2}$, $c\in \mR^{m_2\x
n}$. The
closed-loop system
%\vskip-7mm
\begin{equation}
\label{F}
    F
    :=
        \begin{array}{rl}
         &   \!\!\!\!\!{}_{\leftarrow 2n \rightarrow}{}_{\,\leftarrow m_1 \rightarrow} \\
            \begin{array}{cc}
                {}^{2n} & \hskip-2mm\updownarrow \\
                {}_{p_1} & \hskip-2mm\updownarrow
            \end{array} &\hskip-4mm
            \left[
                \begin{array}{c|c}
                    \ \cA \ & \ \cB \, \\
                    \hline
                    \ \cC\ & \ 0\, \\
                \end{array}
            \right]\\
            {}
        \end{array}    
%    \left[
%        \begin{array}{c|c}
%            \cA & \cB\\
%            \hline
%            \cC & 0
%        \end{array}
%    \right]
    =
    \left[
        \begin{array}{cc|c}
                  a &  bC_2 & bD_{21}\\
                  B_2 c & A & B_1\\
            \hline
                  D_{12} c & C_1 & 0
        \end{array}
    \right],
\end{equation}
governed
by (\ref{nextx})--(\ref{y_to_xi_yxi_to_u}) and depicted in Fig.~\ref{fig:plant},
%==============================================================================
\begin{figure}[htb]
\begin{center}
\unitlength=1.0mm
\begin{picture}(50.00,25.00)
    \put(20,15){\framebox(10,10)[cc]{{\small plant}}}
    \put(40,23){\vector(-1,0){10}}
    \put(20,23){\vector(-1,0){10}}
    \put(42,23){\makebox(0,0)[lc]{$W$}}
    \put(8,23){\makebox(0,0)[rc]{$Z$}}
    \put(40,17){\vector(-1,0){10}}
    \put(20,17){\line(-1,0){10}}
    \put(10,17){\line(0,-1){12}}
    \put(10,5){\vector(1,0){10}}
    \put(30,5){\line(1,0){10}}
    \put(40,5){\line(0,1){12}}
    \put(42,11){\makebox(0,0)[lc]{$U$}}
    \put(8,11){\makebox(0,0)[rc]{$Y$}}
    \put(20,0){\framebox(10,10)[cc]{$K$}}
\end{picture}
\caption{
    The closed-loop system $F$ with input $W$ and
    output $Z$.}
\label{fig:plant}
\end{center}%\vskip-5mm
\end{figure}
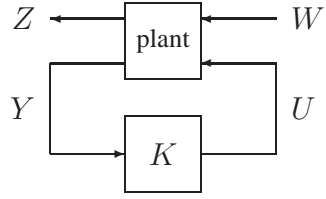
%==============================================================================
has the $2n$-dimensional
combined state $(\Xi,X)$. We formulate a {\it linear
quadro-quartic Gaussian} (LQQG) control problem as the minimization of the
functional (\ref{QQ}) over $n$-dimensional controllers (\ref{K})
such that the matrix $\cA$ of the closed-loop system in
(\ref{F}) is Hurwitz:
\begin{equation}
\label{LQQG}
    \bQ:=\bQ_{\theta}(F)
    =
    -2\Tr(\cA^{\rT}H(I_{2n}+\theta H))
    \longrightarrow
    \min,
    \qquad
    K\
    {\rm stabilizes}\
    F.
\end{equation}
Here, $\theta\>0$ is a given parameter as before, and use is made of Theorem~\ref{th:QQstate}, so that the matrix $H$ is associated by (\ref{H}) with the Gramians $P$, $Q$ of the closed-loop system satisfying  the algebraic Lyapunov equations
\begin{equation}
\label{cPQ}
    \cA P
    +
     P \cA^{\rT}
    +
    \cB\cB^{\rT}
    = 0,
\qquad
    \cA^{\rT} Q
    +
     Q  \cA
    +
    \cC^{\rT}\cC
    = 0.
\end{equation}
In the case $\theta = 0$, the LQQG problem (\ref{LQQG}) reduces to the standard linear quadratic Gaussian (LQG) control problem. For $\theta > 0$, the LQQG problem is a compromise between minimizing the mean value and the variance of the output energy per unit time, with $\theta$ becoming  the relative weight of the quartic norm.

%%%%%%%%%%%%%%%%%%%%%%%%%%%%%%%%%%%%%%%%%%%%%%%%%%%%%%%%%%%%%%%%%%%%%%%%%%%%%%%
\section{Matrices with $\Gamma$-shaped sparsity\label{sec:bGamma}}
%%%%%%%%%%%%%%%%%%%%%%%%%%%%%%%%%%%%%%%%%%%%%%%%%%%%%%%%%%%%%%%%%%%%%%%%%%%%%%%

Since it is convenient to assemble the state-space realization matrices into a matrix with
``$\Gamma$-shaped'' sparsity,  we denote the set of real $(r+p)\x(r+m)$-matrices
with zero bottom-right block of size $(p\x m)$ by
%\vskip-3mm
\begin{equation}
\label{bGamma}
    \bGamma_{r,m,p}
    :=
    \left\{
                 \begin{bmatrix}
                     \rho  &  \sigma \\
                     \tau & 0
                 \end{bmatrix}:\ 
         \rho \in \mR^{r\x r},\ 
         \sigma \in \mR^{r\x m},\ 
         \tau \in \mR^{p\x r}
    \right\}.
\end{equation}
This is a linear subspace of $\mR^{(r+p)\x(r+m)}$ which inherits
the Frobenius inner product of matrices.
Let $\bPi_{r,m,p}$ denote the orthogonal projection onto
$\bGamma_{r,m,p}$ which pads the bottom-right $(p\x m)$-block of a $(r+p)\x (r+m)$-matrix with zeros:
\begin{equation}
\label{bPi}
    \bPi_{r,m,p}
    \left(
                 \begin{bmatrix}
                     \rho  &  \sigma \\
                     \tau & \varpi
                 \end{bmatrix}
    \right)
    =
        \begin{bmatrix}
            \rho  &  \sigma \\
            \tau & 0
        \end{bmatrix}.
\end{equation}
%The subscripts in $\bGamma_{r,m,p}$ and $\bPi_{r,m,p}$ will sometimes be omitted for brevity.
The dependence of the   closed-loop system matrices $\cA$, $\cB$,
$\cC$ on the controller matrices $a$, $b$, $c$ in  (\ref{F}) can be written as
\begin{equation}
\label{ABC}
    \Gamma
     :=
        \begin{bmatrix}
            \cA & \cB\\
            \cC & 0
        \end{bmatrix}
    =
    \Gamma_0
    +
    \Gamma_1
    \gamma
    \Gamma_2,
    \qquad
    \gamma
    :=
        \begin{bmatrix}
            a & b\\
            c & 0
        \end{bmatrix}.
\end{equation}
The affine map $\bGamma_{n,p_2,m_2}\ni \gamma \mapsto \Gamma \in
\bGamma_{2n,m_1,p_1}$ is specified completely  by three
matrices
\begin{equation}
\label{Gamma012}
    \Gamma_0
    :=
        \begin{bmatrix}
            \!\!0_n \, \ 0 \ \ \, 0\\
            0 \ \  A \ \  B_1\\
            \!\!\!0 \ \  C_1 \ \, 0
        \end{bmatrix},
        \qquad
    \Gamma_1
    :=
        \begin{bmatrix}
            I_n & \!\!\!0\\
            0 & \!\!\!B_2\\
            0 & \!\!\!D_{12}
        \end{bmatrix},\qquad
    \Gamma_2
    :=
        \begin{bmatrix}
            \!\!\!\!\!\!\!I_n \  \, 0 \ \  \ 0\\
            0 \ \ C_2 \  D_{21}
        \end{bmatrix},
\end{equation}
where $0_n$ denotes the $(n\x n)$-matrix of zeros.

%%%%%%%%%%%%%%%%%%%%%%%%%%%%%%%%%%%%%%%%%%%%%%%%%%%%%%%%%%%%%%%%%%%%%%%%%%%%%%%
\section{Equations for optimal controller}
%%%%%%%%%%%%%%%%%%%%%%%%%%%%%%%%%%%%%%%%%%%%%%%%%%%%%%%%%%%%%%%%%%%%%%%%%%%%%%%

We now obtain necessary conditions of optimality in the class
(\ref{K})  of  $n$-dimensional stabilizing controllers $K$ for the LQQG problem (\ref{LQQG}). To this end,  we compute the Frechet
derivatives of the quadro-quartic functional
 of the closed-loop system $F$
as a composite function $\gamma \mapsto \Gamma \mapsto \bQ$ of the controller matrices $a$, $b$, $c$ and equate the derivatives to zero.
The differentiation is carried out in two steps: we first consider $\cA$, $\cB$, $\cC$ to be independent variables, and then take into account their dependence on $a$, $b$, $c$.

%====================================================================
\begin{lemma}
\label{lem:dLdGamma} The Frechet derivatives of the quadro-quartic functional $\bQ$ with respect
to the closed-loop system matrices $\cA$, $\cB$, $\cC$, assembled into the matrix $\Gamma$
in (\ref{ABC}), are computed as
\begin{equation}
\label{dLdGamma}
    \d_{\Gamma} \bQ
     :=
        \begin{bmatrix}
             \d_{\cA} \bQ & \d_{\cB} \bQ \\
             \d_{\cC} \bQ & 0
        \end{bmatrix}
    =
     2
        \begin{bmatrix}
             R  & \Omega \cB\\
            \cC \Ups & 0
        \end{bmatrix}.
\end{equation}
Here,
\begin{eqnarray}
\label{Ups}
    \Ups
     & := &
     P  + \theta ( P  H  + \Phi),\\
\label{Omega}
    \Omega
     & := &
     Q  + \theta ( H  Q  + \Psi),\\
\label{R}     R  & := &
     H + \theta (H^2 + Q\Phi + \Psi P),%\\
%
%& = &
%     Q  \Ups + \Omega P  -  H  (I+\theta  H)
\end{eqnarray}
with $P$, $Q$ the Gramians from (\ref{cPQ}); the matrix $H$ is given by (\ref{H}), and $\Phi$, $\Psi$ are the controllability and observability Schattenians of $F$ satisfying the algebraic Lyapunov equations
\begin{equation}
\label{Phi_Psi}
    \cA\Phi + \Phi\cA^{\rT} +  P \cC^{\rT}\cC  P
     = 0,
     \qquad
    \cA^{\rT}\Psi
    +
    \Psi\cA +  Q \cB\cB^{\rT}  Q
     =  0.
\end{equation}
\end{lemma}
%====================================================================

\begin{proof}
By recalling (\ref{QQ}) and applying
Lemmas~\ref{lem:diffH2}, \ref{lem:diffH4} of Appendices~B, C to the closed-loop system $F$, it follows that
$$
    \d_{\Gamma}
    \bQ
    =
    \d_{\Gamma}
    (\|F\|_2^2)
    +
    \theta
    \d_{\Gamma}
    (\|F\|_4^4)/2\\
    = 2
    \begin{bmatrix}
        H & Q\cB\\
        \cC P & 0
    \end{bmatrix}
    +
    2\theta
    \begin{bmatrix}
        H^2 + Q\Phi + \Psi P  & (HQ + \Psi)\cB\\
        \cC (PH+\Phi) & 0
    \end{bmatrix}, 
$$
which, in view of the notations (\ref{Ups})--(\ref{R}), implies (\ref{dLdGamma}).
\end{proof}
%====================================================================

%By combining Lemmas~\ref{lem:PhiPsi} and \ref{lem:dLdGamma}, it
%follows that at a critical point of the Lagrange function $\Lambda$,
%the matrix $ R $ from (\ref{R}) takes the form
%\begin{multline*}
%     R
%    =
%    ( Q  + \theta( H  Q +\Psi)) P
%    +
%     Q
%    ( P  + \theta( P  H +\Phi))\\
%    -
%     H (I+\theta  H )
%    =
%     H (I+\theta  H )
%    +
%    \theta
%    (\Psi P  +  Q  \Phi).
%\end{multline*}
%==============================================================================

The Gramians $ P $, $ Q $ of the closed-loop system and related matrices (that is, $H$, $\Phi$,
$\Psi$, $\Ups$, $\Omega$, $ R $) inherit the four $(n\x n)$-block structure of the
matrix $\cA$ in (\ref{F}). The blocks are numbered as follows:
%\vskip-5mm
\begin{equation}
\label{blocks}
    \cA
    :=
    \begin{array}{cc}
    \!{}_{\leftarrow n \rightarrow}\  {}_{\leftarrow n\rightarrow} &\\
            \begin{bmatrix}
                \cA _{11} & \cA _{12}\\
                \cA _{21} & \cA _{22}
            \end{bmatrix}
    &\!\!\!\!\!
        \begin{matrix}
            \updownarrow\!{}^n\\
            \updownarrow\!{}_n
        \end{matrix}\\
        {}
    \end{array}
    =
    \begin{array}{cc}
    \!{}_{\leftarrow n \rightarrow}\  {}_{\leftarrow n\rightarrow} &\\
            \begin{bmatrix}
                \cA _{\bullet 1} & \cA _{\bullet 2}
            \end{bmatrix}
    &\!\!\!\!\!
            \updownarrow\!{}^{2n}
        \\ {}
    \end{array}
    =
    \begin{array}{cc}
    {}_{\leftarrow 2n \rightarrow}\\
            \begin{bmatrix}
                \cA _{1\bullet}\\
                \cA _{2\bullet}
            \end{bmatrix}
    &\!\!\!\!\!
        \begin{matrix}
            \updownarrow\!{}^n\\
            \updownarrow\!{}_n
        \end{matrix}\\
        {}
    \end{array}.
\end{equation}
In this notation, the $(\cdot)_{11}$ blocks are associated with the
controller state, and the $(\cdot)_{22}$ blocks pertain to the
plant state.

%====================================================================
\begin{lemma}
\label{lem:dLdgamma} The Frechet derivatives of the quadro-quartic functional $\bQ$
of the closed-loop system (\ref{F}) with respect to the
controller matrices $a$, $b$, $c$, assembled into the matrix
$\gamma$ in (\ref{ABC}), are computed as
\begin{equation}
 \label{dLdabc0}
    \d_{\gamma} \bQ
     =
        \begin{bmatrix}
             \d_a \bQ & \d_b \bQ \\
             \d_c \bQ & 0
        \end{bmatrix}
    =
    2
        \begin{bmatrix}
             R _{11}
            &  R _{12}C_2^{\rT} + \Omega_{1\bullet} \cB D_{21}^{\rT}\\
            B_2^{\rT} R _{21} +  D_{12}^{\rT}\cC \Ups_{\bullet 1} & 0
        \end{bmatrix},
\end{equation}
where the matrices $\Ups$, $\Omega$, $R$ are defined by (\ref{Ups})--(\ref{R}).
\end{lemma}
%====================================================================

\begin{proof}
 Since
$\bQ$ is a composite function of $a$, $b$, $c$ which enter this functional
 through the matrices $\cA$,
$\cB$, $\cC$ of the closed-loop system $F$, the chain rule yields
\begin{equation}
\label{dLdgamma}
    \d_{\gamma} \bQ
    =
    (
        \d_{\gamma} \Gamma
    )^{\dagger}
    (
        \d_{\Gamma} \bQ
    )
    =
    \bPi_{n,p_2,m_2}
    (
        \Gamma_1^{\rT}
        \d_{\Gamma} \bQ
        \Gamma_2^{\rT}
    ).
\end{equation}
Here, $(\cdot)^{\dagger}$ denotes the adjoint of a linear operator in the
sense of the Frobenius inner product of matrices, and
$\bPi_{n,p_2,m_2}$ is the orthogonal projection onto the subspace
$\bGamma_{n,p_2,m_2}$ defined by (\ref{bGamma})--(\ref{bPi}).
Indeed, the first variation of the affine map $\Gamma$, defined by
(\ref{ABC}), is $\delta \Gamma = \Gamma_1(\delta \gamma) \Gamma_2$.
Hence,
$
    \delta \bQ
    =
    \Tr
    (
        \d_{\Gamma} \bQ
        \delta\Gamma^{\rT}
    )
    =
    \Tr
    (
        \d_{\Gamma} \bQ
        (\Gamma_1(\delta\gamma) \Gamma_2)^{\rT}
    )
    =
    \Tr
    (
         \Gamma_1^{\rT}
        \d_{\Gamma} \bQ
        \Gamma_2^{\rT}
        \delta\gamma^{\rT}
    )
    =
    \Tr
    (
        \bPi_{n,p_2,m_2}
        (
             \Gamma_1^{\rT}
            \d_{\Gamma} \bQ
            \Gamma_2^{\rT}
        )
        \delta\gamma^{\rT}
    )
$,
which establishes (\ref{dLdgamma}). Substitution of the matrices
$\Gamma_1$ and $\Gamma_2$ from (\ref{Gamma012}) into the right-hand
side of (\ref{dLdgamma}) yields
\begin{align}
\nonumber
    \d_{\gamma} \bQ
%     =
%        \begin{bmatrix}
%             \d_a \bQ  & \d_b \bQ  \\
%             \d_c \bQ  & 0
%        \end{bmatrix}\\
    &=
    \bPi_{n,p_2,m_2}
    \left(
                 \begin{bmatrix}
                     I_n & 0 & 0\\
                     0 & B_2^{\rT} & D_{12}^{\rT}\\
                 \end{bmatrix}\!\!
                 \begin{bmatrix}
                      \d_{\cA} \bQ  & \d_{\cB} \bQ  \\
                      \d_{\cC} \bQ  & 0
                 \end{bmatrix}\!\!
                 \begin{bmatrix}
                     I_n & 0\\
                     0 & C_2^{\rT}\\
                     0 & D_{21}^{\rT}
                 \end{bmatrix}\!
    \right)\\
 \label{dLdabc}
    & =
        \begin{bmatrix}
            (\d_{\cA}\bQ )_{11} & (\d_{\cA} \bQ )_{12}C_2^{\rT} + (\d_{\cB} \bQ )_1 D_{21}^{\rT}\\
            B_2^{\rT}(\d_{\cA} \bQ )_{21} +  D_{12}^{\rT}(\d_{\cC} \bQ )_1 & 0
        \end{bmatrix}.
\end{align}
Here, in view of (\ref{dLdGamma}),
\begin{equation}
\label{dLdB1_dLdC1}
    \d_{\cA}\bQ
    = 2 R,
    \qquad
    (\d_{\cB} \bQ )_1
    =
    2\Omega_{1\bullet} \cB,
    \qquad
    (\d_{\cC} \bQ )_1
    =
    2\cC\Ups_{\bullet 1},
\end{equation}
 and the block numbering (\ref{blocks}) is used. The assertion
(\ref{dLdabc0}) of the lemma now follows from (\ref{dLdabc}) and
(\ref{dLdB1_dLdC1}).
\end{proof}
%====================================================================

 Necessary conditions for  optimality in the class of controllers (\ref{K}) for the LQQG problem (\ref{LQQG}) are now obtained by equating the blocks of the matrix $\d_{\gamma}\bQ$ in (\ref{dLdabc0})  to zero:
\begin{eqnarray}
%\label{full:cA}
%    \cA
%    & := &
%        \begin{bmatrix}
%                  a &  bC_2\\
%                  B_2 c & A
%        \end{bmatrix},\\
%\label{full:cB}
%    \cB
%    & := &
%        \begin{bmatrix}
%                  bD_{21}\\
%                  B_1
%        \end{bmatrix},\\
%\label{full:cC}
%    \cC
%    & := &
%        \begin{bmatrix}
%                  D_{12} c & C_1
%        \end{bmatrix},\\
%\nonumber\\
%\label{full:cP}
%    \cA P
%    +
%     P \cA^{\rT}
%    +
%    \cB\cB^{\rT}
%         & = & 0,\\
%\label{full:cQ}
%    \cA^{\rT} Q
%    +
%     Q  \cA
%    +
%    \cC^{\rT}\cC
%        & = & 0,\\
%\nonumber\\
%\label{full:Phit}
%    \cA\Phi + \Phi\cA^{\rT} +  P \cC^{\rT}\cC  P
%     & = &
%    0,\\
%\label{full:Psit}
%    \cA^{\rT}\Psi
%    +
%    \Psi\cA +  Q \cB\cB^{\rT}  Q
%     & = &  0,\\
%\nonumber\\
%\label{full:Ups}
%    \Ups
%      & := &
%     P  + \theta ( P  H  + \Phi),\\
%\label{full:Omega}
%    \Omega
%      & := &
%     Q  + \theta ( H  Q  + \Psi),\\
%%\nonumber\\
%\label{full:cR}
%     R  & := & H + \theta(H^2+Q  \Phi + \Psi P),\\
%\nonumber\\
\label{full:dLda}
             R _{11} = 0,\\
\label{full:dLdb}
     R _{12}C_2^{\rT} + \Omega_{1\bullet} \cB D_{21}^{\rT} = 0,\\
\label{full:dLdc}
            B_2^{\rT} R _{21} +  D_{12}^{\rT}\cC \Ups_{\bullet 1} =
            0.
\end{eqnarray}
%The transformation
%\begin{equation}
%\label{abc_transform}
%    (a,b,c)\mapsto (\sigma a\sigma^{-1},\sigma b,c\sigma^{-1}),
%\end{equation}
%with nonsingular matrices $\sigma$ of order $n$, yielding equivalent
%state-space realizations of the controller, induces the following
%symmetry group for the equations (\ref{full:cA})--(\ref{full:dLdc}):
%\begin{eqnarray}
%    (\cA,\cB,\cC)
%    & \mapsto &
%    (\Sigma\cA\Sigma^{-1}, \Sigma \cB, \cC\Sigma^{-1}),\\
%     P  &\mapsto &\Sigma P \Sigma^{\rT},\\
%     Q  &\mapsto & \Sigma^{-\rT} Q \Sigma^{-1},\\
%    \Phi &\mapsto &\Sigma\Phi\Sigma^{\rT},\\
%    \Psi &\mapsto & \Sigma^{-\rT}\Psi\Sigma^{-1},\\
%     R  &\mapsto &\Sigma^{-\rT} R \Sigma^{\rT}.
%\end{eqnarray}
%Here, $(\cdot)^{-\rT} := ((\cdot)^{-1})^{\rT}$ is the composition of
%the matrix inverse with the transpose, and
%$$
%    \Sigma
%    :=
%        \begin{bmatrix}
%            \sigma & 0\\
%            0_n & I_n
%        \end{bmatrix}.
%$$

%%%%%%%%%%%%%%%%%%%%%%%%%%%%%%%%%%%%%%%%%%%%%%%%%%%%%%%%%%%%%%%%%%%%%%%%%%%%%%%
\section{Observation-state and state-feedback matrices}
%%%%%%%%%%%%%%%%%%%%%%%%%%%%%%%%%%%%%%%%%%%%%%%%%%%%%%%%%%%%%%%%%%%%%%%%%%%%%%%

%==============================================================================
\begin{lemma}
\label{lem:pos_def_11} Suppose the matrix $D_{21}$ is of full
row rank, and $D_{12}$ is of full column rank. Also, let (\ref{K}) be a stabilizing controller with a minimal state-space realization. Then the top-left blocks of the matrices $ P $, $ Q $ from (\ref{cPQ}) and
$\Ups$, $\Omega$ from (\ref{Ups}), (\ref{Omega}) are all positive definite:
\begin{equation}
\label{pos_def_11}
     P _{11}\succ 0,
    \qquad
     Q _{11}\succ 0,
    \qquad
    \Ups_{11}\succ 0,
    \qquad
    \Omega_{11}\succ 0.
\end{equation}
\end{lemma}
%==============================================================================

\begin{proof}
Since $\theta\>0$,
and the matrices $ PH=PQ  P$, $HQ=Q  P Q$, associated with the Gramians $P$, $Q$, and the Schattenians $\Phi$,
$\Psi$ from (\ref{Phi_Psi}) are all positive semi-definite, then (\ref{Ups})
and (\ref{Omega}) imply that $
    \Ups\succcurlyeq  P $
    and
$
    \Omega\succcurlyeq  Q
$. Hence, the same ordering holds for the top-left blocks of these matrices:
$    \Ups_{11}\succcurlyeq  P_{11}$ and $
    \Omega_{11}\succcurlyeq  Q_{11}$.
Therefore, the last two relations in
(\ref{pos_def_11}) will follow from the first two. We will now
prove that $ P _{11}\succ 0$ under the assumptions that
$D_{21}$ is of full row rank and $(a,b)$ is controllable.
Indeed, $ P _{11}$ is the covariance matrix of the controller state:
\begin{equation}
\label{P11}
     P _{11} = \cov(\xi_t)
     =
     \frac{1}{2\pi}
     \int_{-\infty}^{+\infty}
     g(\omega)\Lambda(\omega)g(\omega)^*
     \rd \omega,
     \qquad
     g(\omega)
     :=
     (i\omega I_n-a)^{-1}b,
\end{equation}
where $
\Lambda(\omega):=h(\omega)h(\omega)^*$
 is the spectral density
associated with the observation signal $Y$ from (\ref{y}), with
$h(\omega)
    :=
    D_{21}
    +
    \begin{bmatrix}
        0 & C_2
    \end{bmatrix}
    (i\omega I_{2n}-\cA)^{-1}\cB
$.
From $\lim_{\omega\to \infty} \Lambda(\omega) = D_{21}D_{21}^{\rT}$, it follows that if $D_{21}$ is of  full row rank, then $\Lambda(\omega)\succ 0$ for all sufficiently  large $\omega$, say $|\omega|> \omega_0$. Now, if $P_{11}$ is singular, then $v^{\rT}P_{11} v = 0$ for some nonzero $v\in \mR^n$. In this case, (\ref{P11}) yields
$
    0=
    v^{\rT}
    P_{11} v
    \>
    (2\pi)^{-1}
    \int_{|\omega|>\omega_0}
    \|g(\omega)^* v\|_{\Lambda(\omega)}^2
    \rd \omega
$,
which, in view of $\Lambda(\omega)\succ 0$ over the high frequency range,   implies that $v^{\rT}g(\omega) = 0$ for all $|\omega|>\omega_0$. Hence, by considering the first $n$ terms of the Laurent series $v^{\rT} g(\omega) = \sum_{k=1}^{+\infty} v^{\rT} a^{k-1} b/(i\omega)^k$ at infinity \cite[Lemma~2.3 on pp. 16--17]{Kimura_1997}, it follows that the rank of the matrix $\left[b \mid  \ldots \mid a^{n-1}b\right]$ is less than $n$, and the pair $(a,b)$ is not controllable. Thus, the full row rank of $D_{21}$ and the controllability of $(a,b)$ indeed ensure $P_{11}\succ 0$.  By duality, a similar
reasoning shows that the observability of $(a,c)$ and the full
column rank condition on $D_{12}$ imply $ Q _{11}\succ 0$.
\end{proof}

%==============================================================================

%==============================================================================
\begin{theorem}
\label{th:bc} Suppose the matrix $D_{21}$ is of full row rank,
and $D_{12}$ is of full column rank. Then the matrices $b$ and $c$ of an optimal controller (\ref{K}) in the LQQG problem (\ref{LQQG}) with a minimal state-space
realization satisfy
\begin{eqnarray}
\label{b}
    b
    =
    -\Omega_{11}^{-1}
    (
         R _{12}C_2^{\rT}
        +
        \Omega_{12} B_1 D_{21}^{\rT}
    )
    (
        D_{21}
        D_{21}^{\rT}
    )^{-1},\\
\label{c}
    c
    =
    -
    (
        D_{12}^{\rT}
        D_{12}
    )^{-1}
    (
        B_2^{\rT}
         R _{21}
        +
        D_{12}^{\rT}C_1\Ups_{21}
    )
    \Ups_{11}^{-1},
\end{eqnarray}
where the matrices $\Ups$, $\Omega$, $R$ are defined by (\ref{Ups})--(\ref{R}).
\end{theorem}
%==============================================================================

\begin{proof}
 Substitution of
the matrices $\cB$ and $\cC$ from (\ref{F})  into (\ref{full:dLdb}) and (\ref{full:dLdc}) brings
these equations to the form
\begin{eqnarray}
\label{full:dLdb1}
     R _{12}C_2^{\rT}
    +
    (
        \Omega_{11}bD_{21}
        +
        \Omega_{12}B_1
    )
    D_{21}^{\rT}
    =
    0,\\
\label{full:dLdc1}
    B_2^{\rT} R _{21}
    +
    D_{12}^{\rT}
    (
        D_{12}c \Ups_{11}
        +
        C_1\Ups_{21}
    )
    =
    0.
\end{eqnarray}
By Lemma~\ref{lem:pos_def_11}, the matrices $\Ups_{11}$ and
$\Omega_{11}$ are nonsingular. Therefore, left multiplication of both
sides of (\ref{full:dLdb1}) by $\Omega_{11}^{-1}$ and right
multiplication by $(D_{21} D_{21}^{\rT})^{-1}$ yields (\ref{b}).
Similarly, right multiplication of both sides of (\ref{full:dLdc1})
by $\Ups_{11}^{-1}$ and left multiplication by $(D_{12}^{\rT}
D_{12})^{-1}$ yields (\ref{c}).
\end{proof}
%==============================================================================

Under the assumptions of Theorem~\ref{th:bc}, the modified set of equations for the state-space realization matrices of an optimal controller in the LQQG problem (\ref{LQQG}) is formed by the algebraic Lyapunov equations (\ref{cPQ}), (\ref{Phi_Psi}) and by the algebraic equations  (\ref{full:dLda}), (\ref{b}), (\ref{c}). In the case
 $\theta = 0$, these equations can be shown to yield the two independent Riccati equations for the standard LQG controller.

% These derivatives all vanish at a critical point of
%the Lagrange function. By equating the Frechet derivatives
%(\ref{dLdP}) and (\ref{dLdQ}) to zero, it follows that
%\begin{eqnarray}
%    \Phi
%    & = &
%     Q  + \theta  Q  P  Q  + \wt{\Phi},\\
%    \Psi
%    & = &
%     P  + \theta  P  Q  P  + \wt{\Psi},
%\end{eqnarray}
%where $\wt{\Phi}$ and $\wt{\Psi}$ are positive-semidefinite real
%symmetric  matrices satisfying the Lyapunov equations
%\begin{eqnarray}
%    \cA^{\rT}\wt{\Phi} + \wt{\Phi}\cA + 2\theta  Q \cB\cB^{\rT}  Q
%    & = & 0,\\
%    \cA\wt{\Psi} + \wt{\Psi}\cA^{\rT} + 2\theta  P \cC^{\rT}\cC  P
%    & = &
%    0.
%\end{eqnarray}
%These representations follow from the lemma below which further
%develops the Gramian anticommutation idea of Remark~2.

%M. Konstantinov, V. MehrmannCorresponding P. Petkov, ''On properties
%of Sylvester and Lyapunov operators'', Linear Algebra and its
%Applications Volume 312, Issues 1-3, 15 June 2000, Pages 35-71

%%%%%%%%%%%%%%%%%%%%%%%%%%%%%%%%%%%%%%%%%%%%%%%%%%%%%%%%%%%%%%%%%%%%%%%%%%%%%%%
\section{Homotopy method}
%%%%%%%%%%%%%%%%%%%%%%%%%%%%%%%%%%%%%%%%%%%%%%%%%%%%%%%%%%%%%%%%%%%%%%%%%%%%%%%

With the matrix $\gamma$ from (\ref{ABC}), we associate a linear subspace of $\bGamma_{n,p_2,m_2}$ by
\begin{equation}
\label{mT}
    \mT(\gamma)
    =
    \left\{
        \begin{bmatrix}
            \tau a - a\tau & \tau b\\
            -c\tau & 0
        \end{bmatrix}:\,
        \tau \in \mR^{n\x n}
    \right\}.
\end{equation}
This is the tangent space  generated by the group of transformations $(a,b,c)\mapsto (\sigma a\sigma^{-1}, \sigma b, c \sigma^{-1})$ (where $\sigma \in \mR^{n\x n}$ are arbitrary nonsingular matrices),  which leave the transfer function of the controller (\ref{K}), and hence, the input-output operator of the closed-loop system (\ref{ABC}), unchanged. The matrix $\d_{\gamma}\bQ$, associated with the  controller $K$, belongs to the orthogonal complement $\mT(\gamma)^{\bot}$ of $\mT(\gamma)$ to $\bGamma_{n,p_2,m_2}$ in the sense of the Frobenius inner product. We say that the controller  delivers a strong local minimum to the quadro-quartic functional $\bQ$ in (\ref{LQQG}) if, in addition to the equality $\d_{\gamma}\bQ = 0$, it also makes the second order Frechet derivative
$
    \d_{\gamma}^2
    \bQ
    =
    \d_{\gamma}^2(\|F\|_2^2)
    +
    \theta
    \d_{\gamma}^2(\|F\|_4^4)/2
$
positive definite on the subspace $\mT(\gamma)^{\bot}$. Now, suppose there exists a smooth map $0\< \theta \mapsto \gamma_*(\theta)\in \bGamma_{n,p_2,m_2}$ such that $\gamma_*(\theta)$ is a strong local minimum  of the quadro-quartic functional $\bQ_{\theta}$ of the closed-loop system $F$ in the sense above, so that
$
    \left.\d_{\gamma}\bQ_{\theta}\right|_{\gamma = \gamma_*(\theta)} = 0
$.
By differentiating the last equality with respect to $\theta$, it follows that
\begin{equation}
\label{homo}
    \d_{\gamma}^2 \bQ_{\theta}(\gamma_*')
    +
    \d_{\gamma}(\|F\|_4^4)/2 = 0.
\end{equation}
Here, $\gamma_*'(\theta):= \d_{\theta}\gamma_*(\theta)$ and use is made of the identity $\d_{\theta} \bQ_{\theta} =\|F\|_4^4/2$ which follows from (\ref{QQ}) and, in view of the interchangeability  of the derivatives in $\theta$ and $\gamma$,  implies that $\d_{\theta}\d_{\gamma}\bQ_{\theta} = \d_{\gamma}(\|F\|_4^4)/2\in \mT(\gamma)^{\bot}$. Since the matrix $\gamma_*(\theta)$ is defined up to the orbit of the transformation group, then $\gamma_*'(\theta):= \d_{\theta}\gamma_*(\theta)$ is defined modulo the subspace $\mT(\gamma_*(\theta))$ from (\ref{mT}). Therefore, (\ref{homo}), which is a linear equation with respect to $\gamma_*'(\theta)$, can be restricted to the subspace $\mT(\gamma_*(\theta))^{\bot}$. As long as $\gamma_*(\theta)$ is a strong local minimum of $\bQ_{\theta}$, so that the self-adjoint operator $\d_{\gamma}^2\bQ$ is positive definite (and hence, invertible) on $\mT(\gamma_*(\theta))^{\bot}$, the equation  (\ref{homo}) is equivalent to
\begin{equation}
\label{homo_ODE}
    \gamma_*'(\theta)
    =
    -\bL^{-1}(\d_{\gamma}(\|F\|_4^4))/2,
\end{equation}
where $\bL$ is the restriction of $\d_{\gamma}^2\bQ$
 to the subspace $\mT(\gamma)^{\bot}$. The equation (\ref{homo_ODE}) is an ODE, with $\theta\>0$ playing the role of fictitious time. The initial value $\gamma_*(0)$ is provided by the state-space realization triple of the standard LQG controller. The computation of an LQQG controller for $\theta>0$ can be carried out by numerically integrating the homotopy ODE (\ref{homo_ODE}) initialized at $\gamma_*(0)$. The operator $\bL$ involves Frechet differentiation of solutions of algebraic Lyapunov equations with respect to their  coefficients, and the inverse $\bL^{-1}$ can be computed by using the vectorization of matrices \cite{M_1988}. The state-space formulae of the homotopy algorithm and other details of its implementation will be reported in subsequent publications.

%We also develop a homotopy method which reduces the solution
%of the set of equations to the numerical integration of the ODE
%\begin{equation}
%%\label{homo_ODE}
%    \cL_{\theta,\gamma}(\gamma\,')
%    +
%    \frac{1}{2}
%    \d_{\gamma}(\|F\|_4^4) = 0,
%    \quad
%    \gamma
%    :=
%    \left[
%        \begin{array}{cc}
%            a & b\\
%            c & 0
%        \end{array}
%    \right],
%\end{equation}
%where $\theta$ plays the role of fictitious time (with the time
%derivative $(\cdot)\,'$) and $\d_{\gamma}$ is the Frechet derivative
%with respect to the block matrix $\gamma$. Here, the second order
%Frechet derivative
%$$
%    \cL_{\theta,\gamma}
%    :=
%    \d_{\gamma}^2
%    \bQ_{\theta}(F)
%    =
%    \d_{\gamma}^2(\|F\|_2^2)
%    +
%    \frac{\theta}{2}
%    \d_{\gamma}^2(\|F\|_4^4)
%$$
%is a self-adjoint linear operator on a Hilbert space of real
%matrices with the sparsity pattern as in (\ref{homo_ODE}) and the
%inherited Frobenius inner product. With equivalent state-space
%realizations being identified, $\gamma$ delivers a strong local
%minimum to the LQQG control problem (\ref{LQQG}) if and only if
%$\cL_{\theta,\gamma}$ is essentially positive definite that can be
%monitored in the course of numerically integrating the homotopy ODE
%(\ref{homo_ODE}).

%%%%%%%%%%%%%%%%%%%%%%%%%%%%%%%%%%%%%%%%%%%%%%%%%%%%%%%%%%%%%%%%%%%%%%%%%%%%%%%

%%%%%%%%%%%%%%%%%%%%%%%%%%%%%%%%%%%%%%%%%%%%%%%%%%%%%%%%%%%%%%%%%%%%%%%%%%%%%%%
\appendix
%%%%%%%%%%%%%%%%%%%%%%%%%%%%%%%%%%%%%%%%%%%%%%%%%%%%%%%%%%%%%%%%%%%%%%%%%%%%%%%

%%%%%%%%%%%%%%%%%%%%%%%%%%%%%%%%%%%%%%%%%%%%%%%%%%%%%%%%%%%%%%%%%%%%%%%%%%%%%%%
\subsection{Covariance of squared norms of Gaussian
random  vectors}
%%%%%%%%%%%%%%%%%%%%%%%%%%%%%%%%%%%%%%%%%%%%%%%%%%%%%%%%%%%%%%%%%%%%%%%%%%%%%%%
\renewcommand{\theequation}{A.\arabic{equation}}
\setcounter{equation}{0}
%%%%%%%%%%%%%%%%%%%%%%%%%%%%%%%%%%%%%%%%%%%%%%%%%%%%%%%%%%%%%%%%%%%%%%%%%%%%%%%

%==============================================================================
\begin{lemma}
\label{lem:bilinear} Let $\xi$ and $\eta$ be jointly Gaussian random
vectors with zero mean. Then the covariance of their squared
Euclidean norms is expressed in terms of the Frobenius norm of their
cross-covariance matrix by
\begin{equation}
\label{bilinear}
    \cov
    (
        |\xi|^2,
        |\eta|^2
    )
    =
    2
    \|
        \cov(\xi,\eta)
    \|^2.
\end{equation}
\end{lemma}
%==============================================================================

\begin{proof}
By applying the representation \cite{Isserlis_1918} for the mixed
 moments of Gaussian random
variables in terms of their covariances to the entries of the
vectors $\xi$ and $\eta$, it follows that
$    \bE(\xi_i^2\eta_j^2)
    =
%    \bE(
%            \xi_i
%            \xi_i
%            \eta_j
%            \eta_j
%    )
%    =
    \bE
    (
        \xi_i\xi_i
    )
    \bE
    (
        \eta_j\eta_j
    )
     +
    \bE
    (
        \xi_i\eta_j
    )
    \bE
    (
        \xi_i \eta_j
    )
     +
    \bE
    (
        \xi_i\eta_j
    )
    \bE
    (
        \xi_i\eta_j
    )
    =
    \bE(\xi_i^2)
    \bE(\eta_j^2)
    +
    2
    (
        \cov
        (
            \xi_i,
            \eta_j
        )
    )^2.
$
%$    \bE(\xi_i^2\eta_j^2)
%    =
%    \bE(
%        \underbrace
%        {
%            \xi_i
%            \xi_i
%            \eta_j
%            \eta_j
%        }_{\tiny \!1 \, \ \ 2\,\ \ 3\,\ \ 4}
%    )
%    =
%    \bE
%    (
%        \underbrace{\xi_i\xi_i}_{\tiny 1\ \ 2}
%    )
%    \bE
%    (
%        \underbrace{\eta_j\eta_j}_{\tiny 3\ \ 4}
%    )
%     +
%    \bE
%    (
%        \underbrace{\xi_i\eta_j}_{\tiny 1\ \ 3}
%    )
%    \bE
%    (
%        \underbrace{\xi_i \eta_j}_{\tiny 2\ \ 4}
%    )
%     +
%    \bE
%    (
%        \underbrace{\xi_i\eta_j}_{1\ \ 4}
%    )
%    \bE
%    (
%        \underbrace{\xi_i\eta_j}_{2\ \ 3}
%    )
%    =
%    \bE(\xi_i^2)
%    \bE(\eta_j^2)
%    +
%    2
%    (
%        \cov
%        (
%            \xi_i,
%            \eta_j
%        )
%    )^2.
%$
%\vskip-8mm\begin{multline*}
%    \bE(\xi_i^2\eta_j^2)
%    =
%    \bE(
%        \overbrace
%        {
%            \xi_i
%            \xi_i
%            \eta_j
%            \eta_j
%        }^{\tiny \!1 \, \ \ 2\,\ \ 3\,\ \ 4}
%    )\\
%    =
%    \bE
%    (
%        \underbrace{\xi_i\xi_i}_{\tiny 1\ \ 2}
%    )
%    \bE
%    (
%        \underbrace{\eta_j\eta_j}_{\tiny 3\ \ 4}
%    )
%     +
%    \bE
%    (
%        \underbrace{\xi_i\eta_j}_{\tiny 1\ \ 3}
%    )
%    \bE
%    (
%        \underbrace{\xi_i \eta_j}_{\tiny 2\ \ 4}
%    )
%     +
%    \bE
%    (
%        \underbrace{\xi_i\eta_j}_{1\ \ 4}
%    )
%    \bE
%    (
%        \underbrace{\xi_i\eta_j}_{2\ \ 3}
%    )\\
%    =
%    \bE(\xi_i^2)
%    \bE(\eta_j^2)
%    +
%    2
%    (
%        \cov
%        (
%            \xi_i,
%            \eta_j
%        )
%    )^2.
%\end{multline*}\vskip-2mm\noindent
Therefore,
\begin{equation}
\label{bilinear1}
    \bE(
        |\xi|^2
        |\eta|^2
    )
     =
    \sum_{i,j}
    \bE(\xi_i^2 \eta_j^2)
     =
    \bE(|\xi|^2) \bE(|\eta|^2)
    +
    2
    \sum_{i,j}
    (
        \cov
        (
            \xi_i,
            \eta_j
        )
    )^2,
\end{equation}
where the rightmost sum is $\|\cov(\xi,\eta)\|^2$. The relation (\ref{bilinear}) is now obtained by
substituting (\ref{bilinear1}) into
$\cov(|\xi|^2,|\eta|^2) := \bE(|\xi|^2|\eta|^2) -
\bE(|\xi|^2)\bE(|\eta|^2)$. Note that (\ref{bilinear}) can also be established by using \cite[Lemma~6.2]{Magnus_1978}.
\end{proof}
%==============================================================================

%%%%%%%%%%%%%%%%%%%%%%%%%%%%%%%%%%%%%%%%%%%%%%%%%%%%%%%%%%%%%%%%%%%%%%%%%%%%%%%
\subsection{State space formula for Frechet derivative of $\cH_2$-norm}
%%%%%%%%%%%%%%%%%%%%%%%%%%%%%%%%%%%%%%%%%%%%%%%%%%%%%%%%%%%%%%%%%%%%%%%%%%%%%%%
\renewcommand{\theequation}{B.\arabic{equation}}
\setcounter{equation}{0}
%%%%%%%%%%%%%%%%%%%%%%%%%%%%%%%%%%%%%%%%%%%%%%%%%%%%%%%%%%%%%%%%%%%%%%%%%%%%%%%

%==============================================================================
\begin{lemma}
\label{lem:diffH2}
The Frechet derivative of the squared $\cH_2$-norm $E:= \|F\|_2^2$ of the system (\ref{FABC}), with $A$ Hurwitz, is computed as
\begin{equation}
\label{dEdGamma}
    \d_{\Gamma} E
    =
    2
    \begin{bmatrix}
        H & Q B\\
        C P & 0
    \end{bmatrix},
    \qquad
    \Gamma:=
    \begin{bmatrix}
        A & B\\
        C  & 0
    \end{bmatrix}.
\end{equation}
Here, the matrix $H$ is associated by (\ref{H}) with the Gramians $P$, $Q$ from
 (\ref{PQ}).
\end{lemma}
%==============================================================================

\begin{proof}
The Frechet derivative $\d_{\Gamma} E $ inherits the block structure of the matrix $\Gamma$:
\begin{equation}
\label{dEdGamma_blocks}
    \d_{\Gamma} E
    =
    \begin{bmatrix}
        \d_A E  & \d_B E \\
        \d_C E  & 0
    \end{bmatrix}.
\end{equation}
We will now compute the blocks of this matrix.  To calculate $\d_A E $, let $B$ and $C$ be fixed. Then the first variation of $ E $ with respect to $A$ is
$
    \delta  E
    =
    \Tr(C^{\rT}C\delta P)
    =
    -\Tr((A^{\rT}Q+QA)\delta P)
    =
    -\Tr(Q(A\delta P +(\delta P)A^{\rT}))
    =
    \Tr(Q((\delta A) P +P\delta A^{\rT}))
    =
    2\Tr(H\delta A^{\rT})
$,
which implies that
\begin{equation}
\label{dEdA}
    \d_A E  = 2H.
\end{equation}
Here, use has also been made of the first variation  of the Lyapunov equation for $P$ with constant $B$ which yields
$        A\delta P + (\delta P) A^{\rT} + (\delta A)P + P\delta A^{\rT} = 0
$.
To compute $\d_B E $, we fix $A$ and $C$. Then the observability Gramian $Q$, which is a function of $A$ and $C$, is also constant, and  the first variation of $ E $ with respect to $B$ is
$
    \delta  E
    =
    \Tr(Q\delta(BB^{\rT}))
    =
    \Tr(Q((\delta B)B^{\rT}+B\delta B^{\rT}))
    =
    2\Tr(Q B\delta B^{\rT})
$,
and hence,
\begin{equation}
\label{dEdB}
    \d_B E  = 2QB.
\end{equation}
The derivative $\d_C E $ is calculated by a similar reasoning. Assuming $A$ and $B$ (and so also the controllability Gramian $P$) to be fixed, the first variation of $ E $ with respect to $C$ is
$
    \delta  E
    =
    \Tr(P\delta(C^{\rT}C))
    =
    \Tr(P((\delta C^{\rT})C+C^{\rT}\delta C))
    =
    2\Tr(C P\delta C^{\rT}),
$
which implies that
\begin{equation}
\label{dEdC}
    \d_C E  = 2C P.
\end{equation}
Substitution of (\ref{dEdA})--(\ref{dEdC}) into (\ref{dEdGamma_blocks}) yields (\ref{dEdGamma}).
\end{proof}

\subsection{Frechet differentiation of quartic norm in state space}
%%%%%%%%%%%%%%%%%%%%%%%%%%%%%%%%%%%%%%%%%%%%%%%%%%%%%%%%%%%%%%%%%%%%%%%%%%%%%%%
\renewcommand{\theequation}{C.\arabic{equation}}
\setcounter{equation}{0}
%%%%%%%%%%%%%%%%%%%%%%%%%%%%%%%%%%%%%%%%%%%%%%%%%%%%%%%%%%%%%%%%%%%%%%%%%%%%%%%

%%==============================================================================
%The matrix $H$ in (\ref{H})  satisfies an
%algebraic Sylvester equation
%\begin{equation}
%\label{Sylvester}
%    HA^{\rT} - A^{\rT} H + QBB^{\rT} - C^{\rT}C P = 0.
%\end{equation}
%This is obtained by left-multiplying the first Lyapunov equation in
%(\ref{PQ}) by $Q$ and right-multiplying the second Lyapunov
%equation by $P$, then taking the difference and recalling (\ref{H}):
%\begin{eqnarray*}
%\nonumber
%    0
%    =
%    Q
%    (
%        A P
%        +
%        PA^{\rT}
%        +
%        BB^{\rT}
%    )
%    -
%    (
%        A^{\rT}Q
%        +
%        QA
%        +
%        C^{\rT}C
%    )
%    P\\
%    =
%    H A^{\rT}
%    -
%    A^{\rT} H
%    +
%    QBB^{\rT}
%    -
%    C^{\rT}C P.
%\end{eqnarray*}
%%The equation (\ref{Sylvester}) specifies the matrix $H$ up to the
%%commutant of the matrix $A^{\rT}$.
%%==============================================================================

%==============================================================================
\begin{lemma}
\label{lem:diffH4} The Frechet derivative of the fourth power $N:= \|F\|_4^4$ of the quartic  norm of the system (\ref{FABC}), with $A$ Hurwitz, is computed as
\begin{equation}
\label{dNdGamma}
    \d_{\Gamma} N
    =
    4
    \left[
    \begin{array}{cc}
        H^2 + Q\Phi + \Psi P & (HQ + \Psi)B\\
        C(PH + \Phi) & 0
    \end{array}
    \right].
\end{equation}
Here, the matrix $H$ is associated by (\ref{H}) with the Gramians $P$, $Q$ from
 (\ref{PQ}), and $\Phi$, $\Psi$ are the Schattenians from (\ref{Schattenian_Phi_Psi}).
\end{lemma}
%==============================================================================
\begin{proof}
We will compute the Frechet derivative of $N$ by using the representation
\begin{equation}
\label{FFF}
    N
    =
    2E_1
    =
    2E_2,
    \qquad
    E_1:= \|F_1\|_2^2,
    \qquad
    E_2:= \|F_2\|_2^2,
\end{equation}
of the $\cH_4$-norm from Lemma~\ref{lem:H4state} in terms of the squared $\cH_2$-norms  of the subsidiary systems
$
    F_1
    :=
    (A,B, B^{\rT} Q)$ and $
    F_2
    :=
    (A,PC^{\rT}, C)
$
as composite functions of the matrices $A$, $B$, $C$.
Since the controllability and observability Gramians of $F_1$ are $P$ and $\Psi$, and the controllability and observability Gramians of $F_2$ are $\Phi$ and $Q$, then application of Lemma~\ref{lem:diffH2} from Appendix~B to the systems $F_1$ and $F_2$ yields
\begin{eqnarray}
\label{diffH2F1}
    \d_{\Gamma_1}E_1
    =
    2
    \begin{bmatrix}
        \Psi P & \Psi B\\
        B^{\rT} H & 0
    \end{bmatrix},
    \qquad
    \Gamma_1
    :=
    \begin{bmatrix}
    A & B\\
    B^{\rT}Q & 0
    \end{bmatrix},\ \ \\
\label{diffH2F2}
    \hskip-5mm \d_{\Gamma_2}E_2
    =
    2
    \begin{bmatrix}
        Q \Phi & HC^{\rT}\\
        C\Phi & 0
    \end{bmatrix},
    \qquad
    \Gamma_2
    :=
    \begin{bmatrix}
    A & PC^{\rT}\\
    C & 0
    \end{bmatrix}.\ \ \,
\end{eqnarray}
Suppose the matrices $A$ and $C$ are fixed and hence, so also is $Q$. Then (\ref{diffH2F1}) implies that the first variation of  $E_1$ with respect to $B$ is
\begin{align}
\nonumber
    \delta E_1
    & =
    2
    \Tr(\Psi B \delta B^{\rT})
    +
    2\Tr(B^{\rT} H \delta (B^{\rT}Q)^{\rT})\\
\nonumber
    & =
    2
    \Tr((\Psi B +
    QH^{\rT} B)\delta B^{\rT})\\
\label{deltaH2F1}
    & =
    2
    \Tr(
        (\Psi + HQ)B
        \delta B^{\rT}
    ),
\end{align}
where the identity $Q H^{\rT}= QPQ = HQ$ has also been used. From (\ref{FFF}) and (\ref{deltaH2F1}), it follows that
\begin{equation}
\label{dH4FdB}
    \d_B N
    =
    4
        (HQ+\Psi)B.
\end{equation}
Suppose the matrices $A$ and $B$ are fixed and hence, so also is $P$. Then (\ref{diffH2F2}) implies that the first variation of  $E_2$ with respect to $C$ is
\begin{align}
\nonumber
    \delta E_2
    & =
    2
    \Tr(HC^{\rT}\delta(PC^{\rT})^{\rT})
    +
    2\Tr(C\Phi\delta C^{\rT})
    \\
\nonumber
    &=
    2
    \Tr((C\Phi +
    CH^{\rT} P)\delta C^{\rT})\\
\label{deltaH2F2}
    & =
    2
    \Tr(
        C(\Phi + PH)
        \delta C^{\rT}
    ),
\end{align}
where the identity $ H^{\rT} P= PQP = PH$ has also been used. From (\ref{FFF}) and (\ref{deltaH2F2}), it follows that
\begin{equation}
\label{dH4FdC}
    \d_C N
    =
    4
        C(PH+\Phi).
\end{equation}
Now, let $B$ and $C$ be constant. Then, in view of (\ref{diffH2F1}), the variation of $E_1$ with respect to $A$ is
\begin{equation}
\label{delta_E1}
    \delta E_1
    =
    2\Tr(\Psi P \delta A^{\rT})
    +
    2
    \Tr(B^{\rT} H \delta(B^{\rT}Q)^{\rT})
    =
    2\Tr(\Psi P \delta A^{\rT})
    +
    2
    \Tr(H^{\rT}BB^{\rT} \delta Q).
\end{equation}
The first variation of the Lyapunov equation for $Q$ in (\ref{PQ}) with $C$ constant yields
$    A^{\rT}\delta Q + (\delta Q) A + (\delta A)^{\rT}Q + Q\delta A = 0
$.
Therefore,
\begin{align}
\nonumber
    \Tr((\delta Q) B B^{\rT}H )
    & =
    -\Tr((\delta Q) (AP+PA^{\rT})H )\\
\nonumber
    &=
    -\Tr((\delta Q) APH)
    -\Tr((\delta Q) PA^{\rT}H )\\
\nonumber
    & =
    \Tr(( A^{\rT}\delta Q +(\delta A)^{\rT}Q + Q\delta A)PH)
    -\Tr((\delta Q) PA^{\rT}H )\\
\nonumber
    & =
    2\Tr(H^2 \delta A^{\rT})
    +
    \Tr(P(HA^{\rT}-A^{\rT}H)\delta Q)\\
\nonumber
    & =
    2\Tr(H^2 \delta A^{\rT})
    +
    \Tr(P(C^{\rT}C P- QBB^{\rT})\delta Q)\\
\nonumber
    & =
    2\Tr(H^2 \delta A^{\rT})
    -
    \Tr((A\Phi + \Phi A^{\rT})\delta Q)
    -
    \Tr(H^{\rT}BB^{\rT}\delta Q)\\
\nonumber
    & =
    2 \Tr((H^2+Q\Phi)\delta A^{\rT})
    -
    \Tr(H^{\rT}BB^{\rT}\delta Q)\\
\label{final}
    & =
    \Tr((H^2+Q\Phi)\delta A^{\rT}).
\end{align}
%==============================================================================
Here, we have also used the definition of the controllability Schattenian $\Phi$ in (\ref{Schattenian_Phi_Psi}), and the identity
$    HA^{\rT} - A^{\rT} H =C^{\rT}C P- QBB^{\rT}
$
which is obtained from
(\ref{PQ}) and (\ref{H}) as
$    0
    =
    Q
    (
        A P
        +
        PA^{\rT}
        +
        BB^{\rT}
    )
    -
    (
        A^{\rT}Q
        +
        QA
        +
        C^{\rT}C
    )
    P
    =
    H A^{\rT}
    -
    A^{\rT} H
    +
    QBB^{\rT}
    -
    C^{\rT}C P
$. Substitution of (\ref{final}) into (\ref{delta_E1}) yields
$\delta E_1 = 2\Tr((H^2 + Q\Phi + \Psi P)\delta A^{\rT}) $, which, in view of (\ref{FFF}), implies that
\begin{equation}
\label{dH4FdA}
    \d_A N
    =
    4(H^2 + Q\Phi + \Psi P).
\end{equation}
The representation (\ref{dNdGamma}) now follows from (\ref{dH4FdB}), (\ref{dH4FdC}) and  (\ref{dH4FdA}). \end{proof}

\end{document}